\documentclass{article}

\author{Andrzej Hanyga\\
{\small ul. Bitwy Warszawskiej 1920r 14/52 \,
02-366 Warszawa, PL}}

\title{Wave propagation in anisotropic viscoelasticity}

\usepackage{amsmath,amssymb,amsthm,upgreek,enumerate}

\newcommand{\e}{\mathrm{e}}
\newcommand{\dd}{\mathrm{d}}
\newcommand{\ii}{\mathrm{i}}
\newcommand{\x}{\mathbf{x}}

\newcommand{\n}{\mathbf{n}}
\renewcommand{\u}{\mathbf{u}}
\renewcommand{\v}{\mathbf{v}}
\newcommand{\w}{\mathbf{w}}
\newcommand{\f}{\mathbf{f}}
\newcommand{\ee}{\mathbf{e}}
\newcommand{\kk}{\mathbf{k}}
\newcommand{\A}{\mathbf{A}}
\newcommand{\B}{\mathbf{B}}
\newcommand{\C}{\mathbf{C}}
\newcommand{\E}{\mathbf{E}}
\newcommand{\I}{\mathbf{I}}
\newcommand{\Q}{\mathbf{Q}}
\newcommand{\K}{\mathbf{K}}
\newcommand{\M}{\mathbf{M}}
\newcommand{\N}{\mathbf{N}}
\renewcommand{\P}{\mathbf{P}}

\newcommand{\G}{\mathbf{G}}

\newcommand{\oo}{\mathrm{o}}
\newcommand{\OO}{\mathrm{O}}

\newcommand{\smartqed}{\hfill $\Box$}

\newtheorem{theorem}{Theorem}[section]
\newtheorem{lemma}[theorem]{Lemma}
\newtheorem{corollary}[theorem]{Corollary}

\newtheorem{definition}[theorem]{Definition}
\newtheorem{remark}{Remark}

\begin{document}

\maketitle

\begin{abstract}
We extend the theory of complete Bernstein functions to matrix-valued functions 
and apply it to
analyze Green's function of an anisotropic multi-dimension\-al linear viscoelastic 
problem. Green's function is given by the superposition of plane waves.
Each plane wave is expressed in terms of matrix-valued attenuation and dispersion functions given in terms of a matrix-valued positive semi-definite Radon measure. More explicit formulae are obtained for 3D isotropic viscoelastic Green's functions. 
As an example of an anisotropic medium the transversely isotropic medium with
a constant symmetry axis is considered.
\end{abstract}
\textbf{Keywords:} {viscoelasticity, anisotropy, completely monotonic, matrix-valued complete
Bernstein function, attenuation, dispersion}

\section*{Notation.}
\begin{small}
\begin{tabular}{lll}
$\mathbb{R}$ & & the set of real numbers \\
$\mathbb{C}$ & & the complex plane \\
$\mathbb{R}^d, \mathbb{C}^d$ & & real, complex $d$-dimensional space\\
$]a,b]$ & $\{ x \in \mathbb{R} \mid a < x \leq b\}$ & \\
$\mathbb{R}_+$ & $]0,\infty]$ & \\
$\Im z$, $\Re z$ & & imaginary, real part of $z$\\
$\theta$ & $\theta(t) = \begin{cases} 1 & t \geq 0\\
 0 & t < 0 \end{cases} $ &
  Heaviside unit step function\\
$f\ast_t g$ & $\int_0^\infty f(s) \, g(t-s)\, \dd s$ & Volterra convolution\\
$\mathcal{S}$ & $\{\n \in \mathbb{R}^d \mid \vert\n\vert = 1\}$ & unit sphere\\
$\v^\top, \A^\top$ & & transpose matrices\\
$\v^\dag, \A^\dag$ & & Hermitian conjugate matrices\\
$\v \cdot \w$ & $\v^\top\, \w$ & scalar product of $\v, \w \in \mathbb{R}^d$ \\
$\I, \; \I_d$ & & unit matrix, $d\times d$ unit matrix\\
$f(x) = \oo_a[g(x)]$ & $\lim_{x\rightarrow a} f(x)/g(x) = 0$ &  \\
$f(x) = \OO_a[g(x)]$  & $0 < \lim_{x\rightarrow a} f(x)/g(x) < \infty$ &\\
\end{tabular}
\end{small}

\section{Introduction.}

There is abundant seismological literature on anisotropic attenuation in rocks and 
in the Earth. Although seismic attenuation anisotropy is often associated with anisotropic permeability, it is 
usually expressed in the framework of linear viscoelasticity. 
Its current models have been developed under ad hoc assumptions that lack a rigorous and systematic mathematical background. Anisotropic attenuation is expected in bio-tissues but 
experimental information is scant. There is however significantly anisotropic attenuation
in fiber-reinforced polymers.

In late 1990s a group of leading seismologists associated with the Anisotropists list wondered how to marry viscoelasticity with anisotropy. 
Very general kinds of anisotropy resulting from various combinations of crystalline anisotropy, fine layering and cracks were contemplated. In order to account for the observed free oscillations of the Moon non-trivial relaxation moduli and creep compliances were postulated in the 30's by Sir Harold Jeffreys and others. With this in view a general anisotropic viscoelastic model of constitutive relations of linear viscoelasticity was 
finally constructed in \cite{HanDuality}. The ultimate interest of seismology lies however in anisotropic and lossy wave propagation. A theory of viscoelastic wave propagation 
addressing this interest is
presented in \cite{Carcione3rd}, cf also \cite{CervPs2005}. The approach in these references is phenomenological and does not ensure consistency with the requirements of viscoelasticity as presented in \cite{HanDuality}. The theory of anisotropic 
viscoelastic wave propagation developed here is rigorously derived here from the constitutive equations.

In the 90's an important objective of ultrasonics was explanation of large asymptotic exponents of logarithmic attenuation in many polymers and bio-materials \cite{Szabo,Szabo2,SzaboWu00,ChenHolm03,ChenHolm04}. In order to explain these observations 
viscoelastic theory was often abandoned and {\em ad hoc} acoustic equations inconsistent with viscoelasticity were constructed \cite{Szabo2,SzaboWu00,ChenHolm03,ChenHolm04,KellyMcGoughMeerschaert08}. 

In this paper I purport to extend the theory of one-dimensional viscoelastic 
wave propagation developed in the papers \cite{SerHan2010,HanWM2013}
and a few follow-up papers. 
In these references I assumed that the underlying model for linear acoustics was linear 
viscoelasticity and derived from this assumption the general form of the attenuation 
function. 
 It is however clear that one-dimensional problems of viscoelasticity should be considered as special cases of three-dimensional problems. This was the second motivation of this paper. It was time to develop a  general three-dimensional framework for wave propagation in linear anisotropic viscoelasticity.

From an experimental point of view viscoelastic effects in linear acoustic wave propagation manifest themselves mainly as wave attenuation. Anisotropy is another important aspect of real media and it has to be taken into account for
a better understanding of wave attenuation. In laboratory anisotropic wave attenuation
can be directly measured by pairs of three-component transducers attached to various faces of the specimen \cite{Papadakis}. It is therefore reasonable to examine and 
calculate attenuation directly instead of focusing on the solutions of
viscoelastic initial-value problems. Viscoelastic anisotropy of muscles
and bones has also been studied by ultrasonic imaging of a low-frequency viscoelastic 
field induced in the specimen by the radiation force of a focused ultrasonic beam
(e.g. \cite{GennissonAl}).

High-frequency behavior of the attenuation function is also relevant for regularity and other properties of viscoelastic Green's functions.

Attenuation and dispersion in one-dimensional viscoelasticity has been studied in my previous papers \cite{SerHan2010,HanWM2013,HanJCA,HanUno,HanDue} under the assumption that the 
relaxation modulus is a locally integrable completely monotonic function. 
I now turn my attention to three-dimensional 
viscoelastic problems and in anisotropic viscoelastic materials under the same hypothesis.
 
Frequency-domain Green's functions in linear viscoelasticity can be expressed in terms of the attenuation and dispersion 
functions as well as an algebraic amplitude factor \cite{HanWM2013}. This representation 
of the Green's functions is convenient for some applications because attenuation is  directly experimentally measurable and is often 
used to determine the material properties of a viscoelastic specimen. It is therefore often
preferable to study the phase speed and attenuation as functions of frequency 
instead of numerical calculation of the wave field. A rigorous and fairly complete theory of attenuation and dispersion functions 
in one-dimensional viscoelastic wave propagation has been developed in 
\cite{SerHan2010,HanWM2013,HanJCA,HanUno,HanDue}.
In \cite{SerHan2010,HanWM2013} the attenuation and dispersion of scalar viscoelastic waves was expressed in terms of integrals with respect to a positive
Radon measure. This result has many important consequences regarding the high- and
low-frequency behavior of the attenuation as well as finite speed of propagation  and 
regularity at the  wavefronts \cite{HanWM2013,HanJCA,HanUno,HanDue}. It allows an in-depth  discussion 
of general properties viscoelastic attenuation and dispersion as functions of frequency as well as their effective calculation for specific models of relaxation or creep. 
General properties of these functions are frequently discussed in materials science and applied acoustics \cite{Szabo,Szabo2,SzaboWu00,ChenHolm03,ChenHolm04,KellyMcGoughMeerschaert08,Mobley2009,HolmSinkus2010,NasholmHolm2011}. 
This method has allowed us to demonstrate that some acoustical
models developed for bio-tissues and polymers are inconsistent with viscoelasticity
(e.g. \cite{Szabo2,KellyMcGoughMeerschaert08}). 

The theory developed in \cite{HanWM2013} depends on 
the assumption that the relaxation modulus is a locally integrable completely monotonic (LICM)
function. This assumption is equivalent to the creep compliance being a Bernstein function 
\cite{Molinari,HanDuality} and entails that the wavenumber is given by a complete Bernstein
function (CBF).
In this case the phase function 
can be expressed as a complete Bernstein function (CBF) $\kappa(p)$ of the positive real variable $p$. The last fact allows an identification of phase speed and attenuation by  analytic continuation of $\kappa(p)$ to the imaginary axis $p = -\ii \omega$ ($\ii$ denotes the imaginary unit) 
and a deep analysis of their properties. 

The Green's function of a multi-dimensional viscoelastic medium can be expressed as a sum
over plane waves. The phases of the plane waves depend on their phase speeds
and attenuation functions. The phase speed and the attenuation are  
directly measurable. 
In an anisotropic medium one is often interested in comparing phase speed and attenuation 
of plane waves propagating in various directions. If we keep the assumption that
the relaxation modulus is a rank-4 tensor-valued completely monotonic function (or, equivalently,  the 
creep compliance is a rank-4 tensor-valued Bernstein function \cite{HanDuality}), then, as we shall 
show here, the phase function of each plane wave is again expressed in terms of a CBF of the Laplace variable
$p \in \mathbb{R}_+$, but it is now a real matrix-valued function $\n\cdot\x\,\K_\n(p)$, with the counterpart $\K_\n(p)$ of $\kappa(p)$, additionally depending on the wavefront normal $\n$. Under the same assumption the matrix-valued inverse phase speed $\C_\n(\omega)$ and the attenuation function $\A_\n(\omega)$ can be determined  
by analytic continuation of $\K_\n(p)$ to the imaginary axis $p = -\ii \omega$. 

 We recall that 
in the case of anisotropic elastic media the phase function is expressed in terms 
of a single real symmetric matrix $\C_\n$ independent of frequency. The eigenvectors 
of $\C_\n$ represent the polarizations of the three modes and their eigenvalues determine the
wavefront speeds of the modes and their phases. Consequently Green's function can be decomposed into a sum of modes. 

In anisotropic viscoelasticity the situation is more complicated. The 
inverse phase speed $\C_\n(\omega)$ is a rank-4 tensor-valued function of frequency and the phase of each plane wave 
involves an additional rank-4 tensor-valued function -- the attenuation function $\A_\n(\omega)$. Both functions are in general rank-2 tensor-valued. For simplicity we shall refer to a fixed  coordinate system and consider them as matrices. The matrices
$\C_\n(\omega)$ and $\A_\n(\omega)$ do not commute and they are not expected to have 
the same eigenvectors. If the modes are defined in terms of the eigenvectors of the
inverse phase speed then the attenuation function results in a coupling
of the modes so defined.

One might be tempted to take an different approach and base the definition of the modes on the eigenvectors and eigenvalues of the real matrix-valued function $\K_\n(p)$ for positive $p$. The  
polarizations of the modes and the associated phase speeds and attenuation functions 
would then be defined in terms of the analytic continuation to the imaginary axis
of $\K_\n(p)$ and its spectral decomposition. Such an analytic continuation 
to the cut complex plane $\mathbb{C}\setminus\, ]-\infty,0]$ is feasible because
$\K_\n(p)$ is a matrix-valued CBF. The eigenvalues of $\K_\n(p)$ would then determine the phases, the phase speeds and the attenuation functions of the corresponding modes. The eigenvectors of $\K_\n(-\ii \omega)$ are also eigenvectors of the function 
$\Q_\n(-\ii \omega)$ which appears in the amplitude factor of each plane wave. Hence 
each plane wave can be decomposed in terms of the eigenvectors of $\K_\n(-\ii \omega)$.

Unfortunately there are obstacles to this approach. 
Apart from some exceptional cases the eigenvectors of $\K_\n(p)$ are not frequency
independent. If such is the case then the associated eigenvalue of $\K_\n(p)$ 
fails to be a CBF. Consequently, the theory developed in \cite{HanWM2013} for scalar
single-mode viscoelastic wave propagation does not apply to such modes. It does however apply to some special modes which have constant
polarization vectors, like the transverse waves in transversely isotropic 
viscoelastic media. The special mode has a scalar attenuation function 
while the attenuation of the remaining coupled modes is jointly represented by
a matrix-valued attenuation function. 

In general it is however necessary to resort to matrix-valued 
CBF functions and a matrix-valued attenuation function.

After a formulation and analysis of the underlying problem (Section~\ref{sec:formulation}) 
matrix-valued CBFs are defined in Section~\ref{sec:CBF}. In the
following sections this concept is applied to 
define the inverse phase speed matrix and the attenuation function.
The theory of the matrix-valued CBFs allows a deeper analysis of the properties
of the attenuation function. 

Section~\ref{sec:TI} shows how viscoelastic effects couple 
the elastic modes so that the concepts of matrix-valued wave attenuation and 
matrix-valued inverse 
phase speed become unavoidable. Isotropic viscoelasticity is however an exception 
because all the polarization vectors are independent of frequency 
(Section~\ref{sec:iso3D}). 

\section{Formulation of the problem and Green's function.}
\label{sec:formulation}

Let $\mathcal{M}_d$ and $\mathcal{M}^\mathbb{C}_d$ denote the set of real 
and complex $d \times d$ matrices, respectively. For any matrix $\B \in \mathcal{M}^\mathbb{C}_d$ we define the real and imaginary parts by the expressions:
\begin{gather}
\Re \B := \frac{1}{2} (\B + \B^\dag)\\
\Im \B := \frac{1}{2 \ii} (\B - \B^\dag)
\end{gather}

It is assumed that the stress-strain constitutive relation is given by the convolution 
\begin{equation}
\mathbf{T}(t,\x) = \mathsf{G}\ast_t \E(t,\x) \equiv \int_0^\infty \mathsf{G}(s)\, \E(t-s,\x)\,\dd s
\end{equation}
where $\mathbf{T}$ denotes the Cauchy stress tensor and $\E$ denotes the strain tensor.
The Cauchy stress and the strain are elements of the space $S$ of real symmetric rank-2 tensors which we 
endow with the scalar product $\langle \E, \C\rangle := E_{ij}\, C_{ij}$. 
The relaxation modulus is a function on $\mathbb{R}_+$ taking values in the space $W$ of 
symmetric operators on $S$. A symmetric operator on $S$ is a linear mapping $\mathsf{H}$ 
from $S$ to 
itself, such that $\langle \E, \mathsf{H}\, \C \rangle = \langle \C, \mathsf{H}\, \E \rangle$.
For notational simplicity we shall fix a coordinate system in the space, identifying it with $\mathbb{R}^d$, and consider the rank-2 tensors over the space as matrices. 

We recall that an infinitely differentiable function $f(t)$ on $\mathbb{R}_+$ is 
completely monotonic (CM) if
$(-1)^n\, \dd^n f(t)/\dd t^n \geq 0$ for $n \in \mathbb{Z}_+ \cup \{0\}$.
It is assumed here that the relaxation modulus is a completely monotonic function in the sense
explained in \cite{HanDuality}, i.e. for every $\E \in S$ the function 
$s \rightarrow \langle \E, \mathsf{G}(s) \, \E\rangle$ is CM. It is also assumed that 
$\mathsf{G}$ is locally integrable near its only possible singularity at 0. This condition is equivalent to the assumption that the integral $\int_0^1 \mathsf{G}(s) \, \dd s $ is convergent. A locally integrable completely monotonic function will be called a LICM function.

It is proved in \cite{HanDuality} that a $W$-valued function 
$\mathsf{G}$ is LICM if and only if there is a positive Radon measure $\nu$ on
$\mathbb{R}_+ \cup \{ 0 \}$ satisfying the inequality 
\begin{equation} \label{eq:nu}
\int_{[0,\infty[} \frac{\nu(\dd r)}{1 + r} < \infty
\end{equation}
and a measurable function $\mathsf{H}: \mathbb{R}_+ \cup \{0\} \rightarrow W$, defined
and bounded by 1 everywhere except perhaps on subset $\mathcal{E}$ of $\mathbb{R}_+ \cup \{ 0 \}$ of zero measure $\nu$,
such that $\mathsf{H}(r)$ a positive semi-definite operator on $S$ for $r \in \mathbb{R}_+
\cup \{0\} \setminus \mathcal{E}$ and such that 
\begin{equation} \label{eq:GCM}
\mathsf{G}(s) = \int_{[0,\infty[} \e^{-s r}\, \mathsf{H}(r) \, \nu(\dd r)
\end{equation}
This statement is a generalization of Bernstein's Theorem \cite{BernsteinFunctions}. 
An operator $\mathsf{A}$ on $S$ is said to be positive semi-definite if 
$\langle \mathbf{E}_1, \mathsf{A}\, \mathbf{E}_2 \rangle \geq 0$ for every 
pair $\mathbf{E}_1,\mathbf{E}_2$.  
The variable $r$ can be viewed as representing the spectrum of inverse relaxation times of the medium.

We also assume that $\mathsf{G}(s)$ has a finite limit at 0
\begin{equation}
\mathsf{G}^0 := \lim_{s\rightarrow 0+} \mathsf{G}(s) 
\end{equation}
This assumption implies that after a finite step of strain $\E(t) = \theta(t)\, \E_0$
the stress jumps to a finite value $\mathsf{G}^0\, \E_0$ before relaxing.

Since the function $\langle \E, \mathsf{G}(s)\, \E\rangle$ is non-increasing and non-negative, it has a non-negative limit
at $s \rightarrow \infty$ for every $\E \in S$. Since $\mathsf{G}(s)$ is
symmetric, this implies that the equilibrium relaxation modulus 
$\mathsf{G}^\infty := \lim_{s\rightarrow \infty} \mathsf{G}(s)$ exists and 
is a positive semi-definite operator on $S$. In one-dimensional viscoelastic media
the inequalities $G^\infty > 0$ and $G^\infty = 0$ define viscoelastic solids and fluids,
respectively, but in the case of a tensor-valued equilibrium relaxation modulus such a 
distinction is unsatisfactory. 

Note that
\begin{gather}
\mathsf{G}^\infty = \nu(\{0\})\, \mathsf{H}(0) \label{eq:yyy} \\
\mathsf{G}^0 = \nu(\{0\})\, \mathsf{H}(0) + \lim_{t\rightarrow 0} \int_{]0,\infty[} 
\e^{-r t}\, \mathsf{H}(r) \, \nu(\dd r) = \mathsf{G}^\infty + 
\int_{]0,\infty[} \mathsf{H}(r) \, \nu(\dd r) \label{eq:lastint}
\end{gather}
Note that $\mathsf{H}(0)$ is always defined if $\nu(\{0\}) > 0$. 
In particular, $\mathsf{G}^\infty = 0$ if $\nu(\{0\}) = 0$, which is the case for viscoelastic fluids. We shall assume for simplicity that
the medium is a viscoelastic solid in the sense that the operator $\mathsf{G}^\infty$ is
positive definite and therefore invertible.\footnote{An anisotropic
viscoelastic medium can behave like a viscoelastic solid for some wavefront normals 
$\n$: $\mathsf{G}_\n > 0$ and as a viscoelastic fluid for other values of $\n$. Furthermore,
viscoelastic solids in a weak sense: $\mathsf{G}_\n \geq 0$ should be considered.
We shall assume that the medium is a viscoelastic solid in the strict sense for all $\n \in \mathcal{S}$. }

Equation~\eqref{eq:nu} does not ensure that the last integral in 
equation~\eqref{eq:lastint} is finite because the relaxation modulus can be singular at 0.
Our assumption that $\mathsf{G}^0$ is finite is equivalent to a stronger 
inequality
\begin{equation}
\nu([0,\infty[) < \infty
\end{equation} 

Concerning equation~\eqref{eq:yyy}, note that $\e^{-s r} < 1/(1 + s r) < 1/(1+ r)$
for $s > 1$. Hence \eqref{eq:yyy} follows from \eqref{eq:nu} and the Lebesgue 
Dominated Convergence Theorem. 

Let 
\begin{equation}
\mathsf{Q}(p) := p \int_0^\infty \e^{-p t}\, \mathsf{G}(t) \, \dd t =
p \int_{[0,\infty[} (p + r)^{-1} \, \mathsf{H}(r)\, \nu(\dd r) 
\end{equation}
and define the matrix-valued function $\Q_\n(p)$ by the formula
\begin{equation}
Q_{\n;ik}(p) = Q_{ijkl}(p)\, n_j\, n_l
\end{equation}
where $Q_{ijkl}(p)$ denotes the operator $\mathsf{Q}(p) \in W$ in a coordinate system on
$\mathbb{R}^d$ and $\n$ is a unit vector in $\mathbb{R}^d$.

The matrix $\Q_\n(p)$ is symmetric for all $p$ in its domain of definition. For real $p$ it is also real and therefore Hermitian. 
There is a natural order relation on the set of Hermitian matrices:
$\A \leq \B$ if $\v^\dag\, \A \, \v \leq \v^\dag\, \B\, \v$ for all $\v \in \mathbb{C}^d$. 
We also recall that
$$\lim_{p\rightarrow 0} \mathsf{Q}(p) = \mathsf{G}^\infty$$
and
$$\lim_{p\rightarrow \infty} \mathsf{Q}(p) = \mathsf{G}^0$$

We shall prove that in Section~\ref{sec:application} that $\Q_\n(p)$ is a matrix-valued 
complete Bernstein function. Hence it is non-decreasing on $\mathbb{R}$ and
$\mathsf{G}_\n^\infty \leq \Q_\n(p) \leq \mathsf{G}_\n^0$
where $\mathsf{G}_\n^0$ and $\mathsf{G}_\n^\infty$ are defined as 
$G^0_{ijkl}\, n_j\, n_l$ and $G^\infty_{ijkl}\, n_j\, n_l$, respectively.
Our assumptions in the beginning of the section imply that
$$0 < \mathsf{G}_\n^\infty \leq \Q_\n(p) \leq \mathsf{G}_\n^0 \qquad \text{for $p \in \mathbb{R}_+\cup \{ 0\}$}$$
hence in particular $\Q_\n(p)$ is invertible for $p \in \mathbb{R}_+ \cup \{ 0\}$. 

The order on the set of real $d\times d$ matrices $\A < \B$ ($\A \leq \B$) is defined 
by the relations $\w^\top\, \A \, \w < \w^\top\, \B \, \w$ for all non-zero vectors
$\w \in \mathbb{R}^d$ ($\w^\top\, \A \, \w \leq \w^\top\, \B \, \w$ for all 
$\w \in \mathbb{R}^d$). 

We shall now prove that the matrix $\Q_\n(-\ii \omega)$ is invertible for real $\omega$. 
Indeed,
$$\Re \Q_\n(-\ii \omega) = \omega^2 \int_{[0,\infty[} \mathsf{H}_\n(r) \frac{\nu(\dd r)}{\omega^2 + r^2}$$
Here
$\mathsf{H}_\n(r)$ is defined in the indicial notation as $H_{ijkl}(r)\, n_j \, n_l$.
The matrix $\mathsf{H}_\n(r)$ is positive semi-definite for $r$ in the support of the measure $\nu$.
We shall exclude the case of $\nu = 0$ and the case of $\mathsf{H}_\n(r) = 0$ for
all $r$ in the support of the measure $\nu$ on $[0,\infty[$ (this would be incompatible 
even with pure elasticity).
Hence $\Re \Q_\n(-\ii \omega) > 0$ and $\Q_\n(-\ii \omega)$ is invertible for $\omega \neq 0$.
For $\omega = 0$ we note that
$$\lim_{\omega\rightarrow 0} \Re \Q_\n(-\ii \omega) = \nu(\{0\})\, \mathsf{H}_\n(0) +
\lim_{\omega\rightarrow 0} \int_{]0,\infty[} \left(1 + r^2/\omega^2\right)^{-1}\, \mathsf{H}_\n(r)\, \nu(\dd r) 
= \mathsf{G}_\n^\infty > 0$$
by the Lebesgue Dominated Convergence Theorem. (Note that $\mathsf{H}_\n$ is $\nu$ integrable because of our 
assumption that $\mathsf{G}^0$ is finite).  This ends the proof.  
Note that $\omega^2/\left(\omega^2 + r^2\right)$ is an non-decreasing function of $\omega$ for
every $r > 0$; hence 
$\Re \Q_\n(-\ii \omega)$ is non-decreasing and $\Re \Q_\n(-\ii \omega) \geq \mathsf{G}_\n^\infty$.

We now calculate Green's function defined as the solution of the initial-value 
problem 
\begin{gather}
\rho \, \partial_t^{\;2} \u = \nabla^\top \mathsf{G}\ast_t \nabla \,\partial_t\, \u, \qquad t > 0,\; \x \in \mathbb{R}^d  \label{eq:IVP1}\\
\u_{,t}(0,\x) = \w(\x) \label{eq:IVP2}
\end{gather}

We shall now define the matrix-valued function 
\begin{equation}
\K_\n(p) := \rho^{1/2} \,p \, \Q_\n(p)^{-1/2}
\end{equation}
where $\Q_{\n}(p)^{1/2}$ denotes the principal square root of $\Q_\n(p)$ for $p$ in the closed right half of the complex plane
$\mathbb{C}_+ := \{ z \mid \Re z \geq 0\}$
(see Appendix~\ref{app:matrices}). 
We note that the eigenvalues of $\Q_{\n}(p)$ for $p \in \mathbb{C}_+$ lie in the open right half of the complex plane. 
Indeed, let $p = q - \ii \omega$,
$q \geq 0$, $\omega \in \mathbb{R}$ and $\v \in \mathbb{R}^d$. Since the function 
$\mathsf{G}$ is CM, equation~\eqref{eq:GCM} 
implies that
\begin{multline*}
\v^\top\, \Q_{\n}(q - \ii \omega)\, \v = (q - \ii \omega) \, \int_{[0,\infty[} 
\frac{\v^\top\, \mathsf{H}_\n(r)\, \v}{q + r - \ii \omega} \nu(\dd r) \\
= \int_{[0,\infty[} \frac{q\, (q + r) + \omega^2 - \ii \,\omega\, r}{(q + r)^2 + \omega^2}   \v^\top\, \mathsf{H}_\n(r)\, \v\, \nu(\dd r), \qquad \v \in \mathbb{R}^d
\end{multline*}
hence $\Re \left[\v^\top \, \Q_\n(q - \ii \omega)\, \v\right] > 0$ for all $\v \in \mathbb{R}^d$, $\v \neq 0$. Hence for $q \geq 0$
the eigenvalues of $\Q_\n(q - \ii \omega)$ lie in the open right complex half-plane
except for the degenerate cases of $\v^\top\,\mathsf{H}_\n(r)\, \v = 0$ for $\nu$-almost all $r \geq 0$ or $\nu = 0$.
By Lemma~\ref{lem:sqr} in the Appendix the matrices of $\Q_\n(p)^{1/2}$ and $\Q_\n(p)^{-1/2}$ exist for $\Re p \geq 0$ except in the two degenerate cases. 

Applying to the initial-value problem (\ref{eq:IVP1}--\ref{eq:IVP2}) the Laplace transformation with respect to time and the Fourier transformation with respect
to the spatial coordinates we obtain the following expression
\begin{multline}
\u(t,\x) =  \frac{1}{(2 \uppi)^{d+1}\, \ii } \times \\ \int_\mathcal{B} \e^{p\, t} \left\{
\int_\mathcal{S} \left[\int \e^{\ii k \n\cdot\x}\, 
\left[ \rho \, p^2 \, \mathbf{I} + k^2\, \Q_\n(p) \right]^{-1}  \, \hat{\w}(k \n) \, 
k^{d-1} \, \dd k \right] \, \Lambda(\dd \n)\right\} \, \dd p
\end{multline}
where $\mathcal{B}$ denotes the Bromwich contour running parallel to the imaginary axis 
in the complex half-plane $\Re p > 0$,
$\mathcal{S}$ denotes the unit sphere $\vert \kk \vert = 1$ in $\mathbb{R}^d$,
$\Lambda(\dd \n)$ is the Lebesgue measure on $\mathcal{S}$ and
$$\hat{\w}(\kk) := \int \e^{\ii \kk \cdot \x} \, \w(\x)\, \dd \x.$$
Hence $\u(t,\x) = \G(t,\x) \ast_\x \w(\x)$, where Green's function $\G$ is given by the formula 
\begin{multline}
\G(t,\x) = \\ \frac{1}{(2 \uppi)^{d+1}\, \ii } \left\{\int_\mathcal{B} \e^{p\, t} 
\int_\mathcal{S} \left[\int_0^\infty \e^{\ii k \n\cdot\x}\, 
\left[ \rho \, p^2 \, \mathbf{I} + k^2\, \Q_\n(p) \right]^{-1}  \, 
k^{d-1} \, \dd k \right] \, \Lambda(\dd \n)\right\}\, \dd p
\end{multline} 
In view of the invertibility of the matrix $\Q_\n(p)$ for $p$ on the imaginary axis,  
\begin{multline}
\G(t,\x) = 
\frac{1}{2 (2 \uppi)^{d+1}\, \ii } \times \\ \int_\mathcal{B} \e^{p\, t} \left\{
\int_\mathcal{S}  \Q_\n(p)^{-1}\,\left[\int_0^\infty \e^{\ii k \n\cdot\x}\, 
 \left[ \K_\n(p)^2 + k^2\right]^{-1} 
k^{d-1} \, \dd k \right] \, \Lambda(\dd \n) \right\}\, \dd p
\end{multline} 

We shall now assume that $d$ is an odd integer. Let $\mathcal{S}_+$ represent the half-sphere 
$\left\{ \n \in \mathcal{S} \mid \n\cdot\x  \geq 0\right\}$. 
\begin{multline} \label{eq:4}
\G(t,\x) = \frac{1}{2 (2 \uppi)^{d+1}\, \ii }  \int_\mathcal{B} \e^{p\, t} \, \Big\{
\int_{\mathcal{S}_+} \Q_\n(p)^{-1}\,\K_\n(p)^{-1} \times \\
\left[\int_{-\infty}^\infty \e^{\ii k \n\cdot\x}\, 
\left\{ \left[ \K_\n(p) - \ii k \, \I\right]^{-1}  + \left[ \K_\n(p) + 
\ii k\, \I\right]^{-1} \right\}\, k^{d-1} \, \dd k\right] \, \Lambda(\dd \n)\,\Big\} \, \dd p = \\
\frac{1}{2 (2 \uppi)^{d+1}\, \ii} \ii^{d-1} \int_\mathcal{B} \e^{p\, t} 
\Big\{\int_{\mathcal{S}_+} \Q_\n(p)^{-1}\,\K_\n(p)^{-1}\times \\
  \left[\int_{\ii \infty}^{-\ii \infty} \e^{-\kappa \n\cdot\x}\, 
 \left[ \K_\n(p) - \kappa \, \I\right]^{-1}  + \left[ \K_\n(p) + \kappa\, \I\right]^{-1} 
\kappa^{d-1} \, \dd \kappa \right] \, \Lambda(\dd \n)\Big\}\, \dd p 
\end{multline}
Let $y := \n\cdot\x$. Green's function will be expressed as a superposition of plane waves
\begin{equation}
\G(t,\x) = \int_\mathcal{S_+} \G_1(t,\n\cdot\x)\, \Lambda(\dd \n)
\end{equation}
where
\begin{multline} \label{eq:G1}
\G_1(t,y) = \frac{1}{2 (2 \uppi)^{d+1}\, \ii} \ii^{d-1}  
\left(-\frac{\partial}{\partial y}\right)^d \int_\mathcal{B} \e^{p\, t} 
 \Q_\n(p)^{-1}\,\K_\n(p)^{-1}\times \\
  \left[\int_{\ii \infty}^{-\ii \infty} \e^{-\kappa \,y}\, 
 \left[\K_\n(p) - \kappa \, \I\right]^{-1}  + \left[ \K_\n(p) + \kappa\, \I\right]^{-1} 
 \, \dd \kappa/\kappa  \right] \,  \dd p 
\end{multline}
The factor $\kappa^{d-1}$ in the integrand has been replaced by a derivative in order 
to deal with an integrand which decays at $\vert \kappa \vert \rightarrow \infty$ 
uniformly with respect to $\arg \kappa \in [-\uppi/2,\uppi/2]$, so that Jordan's lemma \cite{Rudin76} can be applied.

Since $y \geq 0$ we shall close the part of the contour over $\kappa \in [-\ii R, \ii R]$ 
in equation~\eqref{eq:G1} by a half-circle $\vert \kappa \vert = R$ in the half-plane 
$\Re \kappa \geq 0$ and let $R \rightarrow \infty$. The integrand is the product of 
$\OO[1/\kappa]$ and a bounded exponential for $\Re \kappa \geq 0$. 
By Jordan's lemma the integral over the half-circle tends to zero as $R \rightarrow \infty$. 
We can thus consider the contour over
the imaginary axis in the $\kappa$ complex plane to be closed around the residues.

We shall now apply the Cauchy formula for an analytic function $f$ of a complex matrix $\A$
\begin{equation} \label{eq:Cauchy}
f(\A) = \frac{1}{2 \uppi \ii } \int_\Gamma \left[ s \, \I - \A\right]^{-1}\, f(s)\, \dd s
\end{equation}
where the closed contour $\Gamma$ encircles the spectrum of $\A$ exactly once in the positive direction
(\cite{Gantmakher}, Sec.~5.4) (a special case is considered in \cite{Lancaster}, Theorem~5.81).
Equation~\eqref{eq:Cauchy} is often used as a definition of the function of an operator. 
In the case of a matrix it is equivalent
to a definition by a power series if $f$ allows a power series expansion, or in terms of
polynomial interpolation on the spectrum of the matrix. Equation~\eqref{eq:Cauchy} will now be used 
like an extension of the residue calculus with a matrix replacing a pole.

If the matrix $\A$ has the spectral decomposition
$$\A = \sum_{j=1}^d a_j\, \v_j \,\v_j^\dag$$
then 
\begin{equation} \label{eq:sfunspec}
f(\A) = \sum_{j=1}^d f(a_j)\, \v_j \,\v_j^\dag 
\end{equation}
Indeed, for $s \neq a_j$, $j = 1,\ldots,d$, $$(s \I - \A)^{-1} = \sum_{n=1}^d (s - a_j)^{-1} \, \v_j \, \v_j^\dag$$
and \eqref{eq:sfunspec} follows from the Cauchy residue theorem. 

In Section~\ref{sec:application} we shall prove that $\K_\n(p)$ is a complete Bernstein function. 
By Theorem~\ref{thm:specK}, for $\Re p \geq 0$ the spectrum of the matrix $\K_\n(p)$ is contained in the 
open complex right half-plane, hence 
only the first term of the integrand of \eqref{eq:4} contributes, yielding the formula
\begin{equation} \label{eq:res4a}
\G_1(t,y) = \frac{1}{2 (2 \uppi)^d} \ii^{d-1} \, \left(-\frac{\partial}{\partial y}\right)^d \int_\mathcal{B} \e^{p\, t} \,
   \Q_\n(p)^{-1}\,\K_\n(p)^{-2}\, \e^{- y\, \K_\n(p)} \, 
  \dd p 
\end{equation}
This result can be recast in a simpler form by using Theorem~\ref{lem:derexp}:
\begin{equation} \label{eq:res5}
\G_1(t,y) = \frac{1}{2 (2 \uppi)^d} \ii^{d-1} \, \left(-\frac{\partial}{\partial y}\right)^{d-1} \int_\mathcal{B} \e^{p\, t} \,
   \Q_\n(p)^{-1}\,\K_\n(p)^{-1}\, \e^{- y\, \K_\n(p)} \, 
  \dd p 
\end{equation}

\section{Matrix-valued complete Bernstein functions.}
\label{sec:CBF}

\begin{definition} \label{def:CBF}
{\em 
A matrix-valued function} $\A: \mathbb{R}_+ \cup \{ 0 \} \rightarrow \mathcal{M}_d$ 
{\em is said to be a complete Bernstein function (CBF) if 
it has an analytic continuation to the cut complex plane} $\mathbb{C}\setminus\, ]-\infty,0]$ 
{\em which satisfies the inequality}
\begin{equation} \label{eq:0} 
\Im z \, \Im \A(z) \geq 0
\end{equation}
{\em and} $\lim_{x \rightarrow 0+} \A(x)$ {\em exists and is real}.
\end{definition}
\noindent(cf Theorem~6.2 in \cite{BernsteinFunctions} for scalar CBFs).
By definition $\Im \A(z)$ is Hermitian, hence the inequality in \eqref{eq:0} 
makes sense. 

\begin{definition} \label{def:Stieltjes}
{\em A matrix-valued function} $\A: \mathbb{R}_+ \cup \{ 0 \} \rightarrow \mathcal{M}_d$ 
{\em is said to be a Stieltjes function if 
it has an analytic continuation to} $\mathbb{C}\setminus ]-\infty,0]$ 
{\em which satisfies the inequality}
\begin{equation} \label{eq:1} 
\Im z \, \Im \A(z) \leq 0
\end{equation}
{\em and} $\lim_{x \rightarrow 0+} \A(x)$ {\em exists and is real positive semi-definite}.
\end{definition}
\noindent(cf Corollary~7.4 in \cite{BernsteinFunctions}).

If $\B \in \mathcal{M}^\mathbb{C}_d$ is invertible and $\Im \B \geq 0$ then
$$\Im \B^{-1} = \frac{1}{2 \uppi \ii} \left(\B^{-1} - \B^{\dag -1}\right) =
\frac{1}{2 \uppi \ii} \B^{-1}\, \left(\B^\dag - \B\right)\, \B^{\dag -1} \leq 0$$
This proves the following lemma:
\begin{lemma} \label{lem:inverse}
If $\A(x)$ is a matrix-valued CBF and $\A(z)$ is invertible for $z \not\in ]-\infty,0]$
then $\A(x)^{-1}$ is a Stieltjes function.
\end{lemma}

\section{Application of the matrix-valued CBF theory to Green's function.}
\label{sec:application}

\begin{theorem} \label{thm:K}
$\K_\n(\cdot)$ is a CBF for every $\n \in \mathcal{S}$.
\end{theorem}
\begin{proof}
$\Q_\n(p)$ is a CBF, hence $\Q_\n(p)^{1/2}$ is a CBF (Corollary~\ref{cor:half}). By 
Corollary~\ref{cor:toS}
$p^{-1}\, \Q_\n(p)^{1/2}$ is a matrix-valued Stieltjes function.
By Lemma~\ref{lem:inverse}  its inverse \\ $p\, \Q_\n(p)^{-1/2}$ is a CBF. \smartqed
\end{proof} 

Since $\K_\n(0) = 0$, Theorem~\ref{thm:CBF} in the Appendix implies that 
\begin{equation} \label{eq:6a}
\K_\n(p) = p \, \B_\n + p \int_{]0,\infty[} (p + r)^{-1} \, \M_\n(r)\, \mu(\dd r)
\end{equation}
where $\mu$ is a positive Radon measure satisfying the inequality
\begin{equation} \label{eq:6}
\int_{]0,\infty[} (1 + r)^{-1} \,\mu(\dd r) < \infty,
\end{equation}
$\M_\n(r)$ is an $S$-valued function bounded $\mu$-almost everywhere by 1
and $\B_\n \in S$ is positive semi-definite. The subscript $\n$ indicates the parametric
dependence on $\n$. 
Corollary~\ref{cor:coeffs}  implies that 
$$\B_\n = \lim_{p \rightarrow \infty} \left[p^{-1} \, \K_\n(p)\right] = 
\rho^{1/2}\, \left[\lim_{p \rightarrow \infty} \Q_\n(p)\right]^{-1/2} =
\left[\rho^{-1} \,\G^0_\n\right]^{-1/2}$$
where $p$ is considered a real variable, ${G^0_{\n}}_{\,ij} := G^0_{\;ikjl} \, n_k\, n_l$,
is a real symmetric matrix. It can be diagonalized and its eigenvalues 
are non-negative. A non-zero eigenvalue $1/c^{(0)}_\n$ of $\B_\n$ represents 
the inverse wavefront speed for the corresponding mode and a plane wave front
orthogonal to the vector $\n$. 

Given $\K_\n(p)$ as an analytic function the measure $\mu$ can be calculated using
Corollary~\ref{cor:calc}. 

\begin{theorem} \label{thm:specK}
Assume that $\K_\n(p)$ does not vanish. 

For $\Re p \geq 0$ the spectrum of the matrix $\K_\n(p)$ lies in the open right half of the complex $p$-plane. 
\end{theorem}
\begin{proof}
It is sufficient to consider only the second term on the right of \eqref{eq:6a}. If $\Re p \geq 0$ then 
\begin{multline*}
p  \int_{]0,\infty[} (p + r)^{-1} \, \M_\n(r)\, \mu(\dd r) =   \int_{]0,\infty[} \vert p + r\vert^{-2} (\vert p \vert^2 + r \, p) \, \M_\n(r)\, \mu(\dd r)
\end{multline*}
is positive semi-definite. Hence 
$$\langle \v, \K_\n(p)\, \v \rangle > 0 \qquad \forall \v \in \mathbb{R}^d, \v \neq 0$$
which implies the thesis. \smartqed
\end{proof} 

The attenuation function 
\begin{equation} \label{eq:An}
\A_\n(\omega) := \Re \K_\n(-\ii \omega) = 
\omega^2 \int_{]0,\infty[} \left(\omega^2 + r^2\right)^{-1}\,\M_\n(r)\, \mu(\dd r), \qquad \omega \in \mathbb{R}
\end{equation}
is positive semi-definite. The matrix $\A_\n(\omega)$ is real and symmetric for $\omega \in \mathbb{R}$. 

The functions $\omega \rightarrow \omega^2/\left(\omega^2 + r^2\right)$ are increasing
for every $r > 0$ while $\mathsf{H}_\n(r)$ is positive semi-definite. Hence
for every $\v \in \mathbb{R}^d$ the function $\v^\top\, \A_\n(\omega)\, \v$ is non-decreasing.
It follows from Lemma~\ref{lem:est} below that $\lim_{\omega \rightarrow \infty} 
\left[\omega^{-1}\, \A_\n(\omega)\right] = 0$, hence the rate of increase of the attenuation function is sublinear.

The inverse phase speed is defined by the formula 
\begin{equation} \label{eq:Cn}
\mathbf{C}_\n(\omega) := \Re [\K_\n(-\ii \omega)/(-\ii \omega)] \equiv \B_\n + \mathbf{D}_\n(\omega)
\end{equation}
where 
\begin{equation} \label{eq:Dn}
\mathbf{D}_\n(\omega) := \int_{]0,\infty[} \left(\omega^2 + r^2\right)^{-1}\, r\,
\M_\n(r)\, \mu(\dd r) 
\end{equation}
is real, symmetric and positive semi-definite for $\omega > 0$. The inverse phase speed $\mathbf{C}_\n(\omega)$ 
is a generalization of $c(\omega)^{-1}$ in the one-dimensional viscoelasticity, where 
$c(\omega)$ denotes the phase speed. 
Clearly $\mathbf{C}_\n(\omega) \geq \B_\n$, which corresponds to the inequality 
$c(\omega) \leq c_0$ in \cite{HanWM2013}, $c_0$ being the wavefront speed.
By an estimate used in the proof of Lemma~\ref{lem:est} and \eqref{eq:6} it can be shown that 
$$\lim_{\omega\rightarrow\infty} \mathbf{D}_\n(\omega) = 0$$
which implies that 
\begin{equation} \label{eq:lim}
\lim_{\omega\rightarrow\infty} \mathbf{C}_\n(\omega) = \B_\n
\end{equation}
Equation~\eqref{eq:lim} corresponds to the relation $\lim_{\omega\rightarrow\infty} c(\omega) = c_0$
in \cite{HanWM2013}. 

If 
\begin{equation} \label{eq:z}
\int_{]0,\infty[} \mu(\dd r)/r < \infty
\end{equation}
 then, by the Lebesgue Dominated Convergence Theorem, 
$\lim_{\omega\rightarrow 0} \mathbf{D}_\n(\omega) = \mathbf{D}_\n^0$, where
\begin{equation} \label{eq:y}
\mathbf{D}_\n^0 := \int_{]0,\infty[} r^{-1} \M_\n(r) \, \mu(\dd r),
\end{equation}
otherwise $\lim_{\omega\rightarrow 0} \mathbf{D}_\n(\omega)$ diverges to infinity.
Consequently $\lim_{\omega\rightarrow 0} \mathbf{C}_\n(\omega) = \B_\n +
\mathbf{D}_\n^0 =: \mathbf{C}_\n^0$ in the first case (phase speeds are bounded away from zero) 
or $\mathbf{C}_\n(\omega)$ tends to infinity for $\omega \rightarrow 0$.
For each $\w \in \mathbb{R}^3$ the function $\w^\top\, \mathbf{D}_\n(\omega)\, \w$
decreases monotonely, hence $\w^\top \, \C_\n(\omega)\, \w$ also varies monotonely 
between its limits $\w^\top\, \B_\n\, \w$ and $\w^\top\, \mathbf{C}_\n^0\, \w$ (or infinity).

If inequality~\eqref{eq:z} holds then
\begin{equation} \label{eq:zz}
\lim_{p\rightarrow 0} \left[p^{-1}\, \K_\n(p)\right] = \mathbf{C}^0_\n = 
\left[\rho^{-1}\, \mathsf{G}^\infty_\n\right]^{-1/2}
\end{equation}
for all $p \in \mathbb{C}$. In particular this implies 
that $\mathsf{G}^\infty_\n$ is invertible and positive definite. The opposite 
implication is also true,
hence inequality~\eqref{eq:z} is equivalent to the assumption 
that the medium is a viscoelastic solid.

In a viscoelastic solid equation~\eqref{eq:z} implies that 
$$\lim_{\omega\rightarrow 0} \omega^{-1}\, \A_\n(\omega) 
\equiv \lim_{\omega \rightarrow 0} \Im \left[(-\ii \omega)^{-1} \K_\n(-\ii \omega)\right] 
= 0$$
because the right-hand side of \eqref{eq:zz} is Hermitian. In particular, if 
$\A_\n(\omega) \sim_0 a\, \omega^\alpha$ and the medium is a viscoelastic solid, then
$\alpha > 1$. This result is consistent with experimental results for 
bio-tissues, polymers and castor oil \cite{Szabo}.  

Assuming regular variation of $\mu(\dd r)$ and $\M_\n(r)$ various asymptotic estimates
for the attenuation function obtained in \cite{HanWM2013} can be extended to the 
anisotropic three-dimensional case. This would however require a careful extension of the 
regular variation theory to matrix-valued functions \cite{MarkRegVar} and is thus beyond the scope of this paper. High-frequency asymptotics of the attenuation
function is relevant for regularity of the plane wave at the wavefront \cite{HanJCA}.
Wavefront regularity is in turn relevant for the pedestal effect (delay of a signal with respect to the wavefront
\cite{Strick1:ConstQ}) and for travel time inversion in seismology \cite{FastTrack}. In an anisotropic medium regularity
is expected to depend on the direction of propagation. Low-frequency asymptotics 
is relevant for experimental observations of ultrasound in polymers, bio-tissues \cite{Szabo} and in
seismology \cite{HanDue}. 

\section{Plane waves: wavefronts and attenuation.}
\label{sec:plane}

Green's function $\G(t,\x)$ is a superposition of plane waves $\G_1(t,\n\cdot\x)$
with wavefront normals $\n \in \mathcal{S}_+$, 
where $\G_1(t,y) = (-\partial/\partial y)^{d-1} \, \G_2(t,y)$ and 
\begin{equation} \label{eq:planewave}
\G_2(t,y) = \frac{1}{2 (2 \uppi)^d} \ii^{d-1} \int_\mathcal{B} \e^{p\, t} \,
  \Q_\n(p)^{-1}\, \K_\n(p)^{-1} \,\e^{-y\, \K_\n(p)} \, 
 \dd p,  
\end{equation}

We shall now show that the plane waves are exponentially attenuated. Define 
the matrix norm
\begin{equation}
\| \A \|_d := \sup_{\w \in \mathbb{C}^d} \sqrt{\frac{\w^\dag\, \A^\dag\, \A \w}{\w^\dag \w}}
\end{equation} 
The norm $\| \A \|_d$ is unitarily invariant, i. e. for every unitary matrix $\mathbf{U}$
and $\mathbf{V}$ 
$$ \| \mathbf{U} \, \A\, \mathbf{V}\|_d = \| \A \|_d.$$
We can now show that the function 
\begin{equation} \label{eq:Z} 
\e^{-y \K_\n(\omega)} \equiv \e^{y \,\ii\,\omega \,\mathbf{C}_\n(\omega) - y\, \A_\n(\omega)}
\end{equation}
is exponentially attenuated if $\A_\n(\omega) > 0$ for $\omega \in \mathbb{R}_+$.

According to Theorem~IX.3.11 in \cite{Bhatia97} every complex $d \times d$ matrix $\A$
satisfies the inequality 
$$\| \e^{\A} \|_d \leq \| \e^{\Re \A}\|_d$$
Hence 
$$\| \e^{-y \, \K_\n(-\ii \omega)} \|_3 \leq \| \e^{-y \, \A_\n(\omega)}\|_3$$
Let $a_0(\omega)$ be the smallest eigenvalue of the real symmetric matrix $\A_\n(\omega)$.
If $a_0$ is the smallest eigenvalue of a real symmetric
$d \times d$ matrix $\A$ with eigenvalues $a_j$ and unitary eigenvectors $\v_j$, then 
$$\| \e^{-\A}\|_d = \sup_{\{c_j\mid j=1,\ldots d\}} \sqrt{\frac{\sum_{j=1}^d \vert c_j\vert^2 \, \exp(-2 a_j)}{\sum_{j=1}^d \vert c_j\vert^2}} = \e^{-a_0}$$
Consequently
\begin{equation} \label{eq:attn}
\| \e^{-y \, \K_\n(-\ii \omega)} \|_3 \leq \e^{-a_0(\omega) y}
\end{equation}
If $\A_\n(\omega) > 0$ for $\omega \geq 0$ then equation~\eqref{eq:attn} implies that
the plane wave with the wavefront normal $\n$ is exponentially attenuated with
distance. 

We now show that the plane wave vanishes beyond a wavefront propagating with a finite 
speed if $\mathbf{B}_\n > 0$. 

\begin{lemma} \label{lem:est}
$\K_\n(p) = p\, \B_\n + \K^{(1)}_n(p)$
where $\K^{(1)}_\n(p) = \oo[p]$ for $\Re p \geq 0$ and $\vert p \vert \rightarrow \infty$.
\end{lemma}
\begin{proof}
Since $\| \M_\n(r)\| \leq 1$ almost everywhere with respect to $\mu$, 
the integral on the right-hand side of equation~\eqref{eq:6a} is majorized by 
$$\int_{]0,\infty[} \frac{1}{\sqrt{r^2 + \vert p \vert^2 + 2 r\, \Re p}} \mu(\dd r)
\leq \int_{]0,\infty[} \frac{1}{\sqrt{r^2 + \vert p \vert^2}} \mu(\dd r)$$

We shall split the last integral into an integral over $[0,2]$ and over $]2,\infty[$. 
The first integral
$$\int_{]0,2]} \frac{1}{\sqrt{r^2 + \vert p \vert^2}} \mu(\dd r)$$
tends to 0 
as $\vert p \vert \rightarrow \infty$ by the Lebesgue Dominated Convergence Theorem.
It remains to consider the integral 
$$\int_{]2,\infty[} \frac{1}{\sqrt{r^2 + \vert p \vert^2}} \mu(\dd r).$$
For $\vert p \vert > 1$ the integrand is majorized by 
$1/\sqrt{r^2 + 1}$. Since $r \geq 2$, this expression is $\leq \sqrt{2}/(r + 1)$,
which is integrable with respect to the measure $\mu$. By the
Lebesgue Dominated Convergence Theorem 
$$\int_{]2,\infty[} \frac{1}{\sqrt{r^2 + \vert p \vert^2}} \mu(\dd r) \rightarrow 0$$
as $\vert p \vert \rightarrow \infty$. Hence the integral in 
\eqref{eq:6a} tends to 0 as $\vert \omega \vert \rightarrow \infty$,
which implies the thesis.\smartqed
\end{proof} 

The matrix $\B_\n$ is real symmetric. Hence it can be expressed in spectral form
\begin{equation}
\B_\n = \sum_{j=1}^d b_j(\n)\, \v_j \, \v_j^\top
\end{equation}
where $\v_j \in \mathbb{R}^d$, $j = 1,2,\ldots d$ are orthonormal. 
Let $b^0(\n) := \min\{ b_j(\n) \mid j = 1,2, \ldots d\}$.

\begin{theorem}
$\G_1(t,\n\cdot\x) = 0$ for $t < b^0(\n)\, \vert \n\cdot\x\vert$.
\end{theorem} 
\begin{proof}
Consider the integral in \eqref{eq:planewave}. 
Let $t < \vert y\vert \, b^0(\n)$. 
Close the segment $[-\ii R,\ii R]$ of the Bromwich contour by a large half-circle $\vert p \vert = R$, $\Re p \geq 0$
and let $R$ tend to infinity. The only singularities of the integrand are 
the branching cuts of the CBFs $\K_\n(p)$ and $\Q_\n(p)$ along the negative half-axis
and therefore they lie outside the closed contour. 
In the half-plane $\Re p > 0$
\begin{multline*}
\lim_{\vert p \vert \rightarrow \infty} \exp(p\, (t - \vert y\vert\,\B_\n 
-p^{-1}\, \K^{(1)}_\n(p))) = \lim_{\vert p \vert \rightarrow \infty} \exp(p\, (t - \vert y\vert\,\B_\n) = \\ =
\sum_{j=1}^d \lim_{\vert p \vert \rightarrow \infty} \exp(p\, (t - b_j(\n)\, \vert y\vert)
\v_j\, \v_j^\top = 0 
\end{multline*}
because all the exponents have negative real parts. For $\Re p \geq 0$ the exponentials on the right-hand side are bounded functions of $p$.
The integrand of \eqref{eq:planewave} is the product of the exponentials just estimated and 
$\Q_\n(p)^{-1}\, \K_\n(p)^{-1} = \OO[1/p]$ for $p \rightarrow \infty$. The asymptotic estimates are uniform with
respect to $\arg p \in [-\uppi/2,\uppi/2]$ hence, by Jordan's lemma, the integral over the 
half-circle tends to 0 as $R \rightarrow \infty$ while the integral over $[-\ii R, \ii R]$ 
tends to $\G_2(t,y) = 0$. Hence $\G_1(t,y)$ vanishes for $t < b^0(\n)\, \vert y \vert$ 
which implies the thesis. \smartqed
\end{proof} 

Consequently $1/b^0(\n)$ is an upper limit on the propagation speed of a plane-wave front 
with the normal $\n$.

The necessity of evaluating the exponential function of non-Hermitian matrices is 
somewhat discouraging. 
For sufficiently anisotropic small dispersion $\C_\n(\omega) - w(\omega) \, \I_3 = \OO[\varepsilon]$ 
and attenuation $\A_\n(\omega) = z(\omega) \OO[\varepsilon]$, where $w(\omega)$ and $z(\omega)$ are some real functions, the inverse phase function and the attenuation function in \eqref{eq:Z} can be approximately disentangled in the form 
\begin{multline} \label{eq:OUN}
\exp(\ii \omega \C_\n(\omega)\, r - \A_\n(\omega) \, r) = \\
\Phi(\varepsilon\, r, \omega) \,\exp(\ii \omega \C_\n(\omega)\, r) \, \exp(-\A_\n(\omega)\, r)
= \\ \Psi(\varepsilon\, r, \omega)\, \exp(-\A_\n(\omega)\, r) \, \exp(\ii \omega \C_\n(\omega) r)
\end{multline}
where the functions $\Phi(y, \omega), \Psi(y, \omega) = 1 + \OO\left[y^2\right]$ can be expressed in terms of exponential functions of nested commutators of $\C_\n(\omega)$
and $\A_\n(\omega)$
by the use of the Zassenhaus formula \cite{Magnus54}, cf also \cite{CasasAl} for the numerical implementation of the Zassenhaus formula. Substituting the inverse phase function from equations~(\ref{eq:Cn},\ref{eq:Dn}) and the attenuation function from \eqref{eq:An} the functions $\Phi$ and $\Psi$ can be expressed in terms of 
an infinite product of exponentials of nested commutators of the matrices $\ii r \C_\n(\omega)$ and $r \A_\n(\omega)$. 
The matrix $\A_\n(\omega)$ is symmetric 
hence it admits a spectral representation with eigenvalues $a_i(\omega)$ and eigenvectors $\w_i(\omega)$. 
If $\v = \sum_{i=1}^3 v_i\, \w_i(\omega)$ then the the logarithmic attenuation rates 
can be made explicit:
$$\exp(r [\ii \omega \C_\n(\omega) - \A_\n(\omega)]) \, \v = \Phi(r, \omega)\, \exp(\ii \omega r 
\C_\n(\omega)) \sum_{i=1}^3 \exp(-r \, a_{\n,i}(\omega)) \, v_i \, \w_i(\omega)$$

In seismological applications its is more common to identify the wave modes by 
associating them with elastic wave modes. This can be achieved by using the last line of equation~\eqref{eq:OUN} and expanding $\v$ in terms of eigenvectors $\u_i$, $i = 1,2,3$,
of the matrix $\C_\n(\omega)$:
$$\exp(r [\ii \omega \C_\n(\omega) - \A_\n(\omega)]) \, \v = \Psi(r, \omega)\, \exp(-r 
\A_\n(\omega)) \sum_{i=1}^3 \exp(\ii \omega r/c_{\n,i}(\omega)) \, v_i \, \u_i(\omega)$$
where $c_{\n,i}$ are the phase speeds of the modes. 
Note that the modes are coupled by the attenuation function. From the last representation it is easy to see that the wavefront propagation speeds of
the modes are given by $c_{\n,i}^\infty := \lim_{\omega\rightarrow \infty} c_{\n,i}(\omega)$.

\section{The isotropic three-dimensional case}
\label{sec:iso3D}

We now assume that $d = 3$ and
$$\G_\n(t) = (\lambda(t) + 2 \mu(t))  \, \P_\parallel(\n) + \mu(t)\, \P_\perp(\n)$$
where $\lambda(t), \mu(t)$ are LICM functions and 
$\P_\parallel(\n) := \n \, \n^\top$, $\P_\perp(\n) = \I - \P_\parallel(\n)$ for every $\n \in \mathcal{S}$. It is easy to conclude therefrom that
the relaxation modulus $\mathsf{G}$ is a tensor-valued LICM function.

In this case
$$\Q_\n(p)^{-1/2}  = (\overline{\lambda}(p) + 2 \overline{\mu}(p))^{-1/2} \, \P_\parallel(\n) + 
\overline{\mu}(p)^{-1/2}\, \P_\perp(\n)$$
where $\overline{\lambda}(p) := p \, \tilde{\lambda}(p)$ and $\overline{\mu}(p) := p \, \tilde{\mu}(p)$
are CBFs.
Furthermore
$$\K_\n(p) = \kappa_\parallel(p) \, \P_\parallel(\n) + \kappa_\perp(p)\, \P_\perp(\n),$$
where 
\begin{gather}
\kappa_\parallel(p) := \rho^{1/2} p/(\overline{\lambda}(p) + 2 \overline{\mu}(p))^{1/2}\\
\kappa_\perp(p) := \rho^{1/2} p/\overline{\mu}(p)^{1/2}
\end{gather}
are again CBFs \cite{HanWM2013}. Applying equation~\eqref{eq:sfunspec} we obtain the formula 
\begin{multline*}
\G(t,\x) = 
-\frac{1}{2 (2 \uppi)^3\, \rho} \times \\ \frac{\partial^2}{\partial y^2} \int_\mathcal{B} \frac{\dd p}{p^2} 
\int_{\mathcal{S}_+} \left[ \kappa_\parallel(p)\,  \e^{-y \,\kappa_\parallel(p)} \, \P_\parallel(\n) + 
\kappa_\perp(p) \,\e^{-y \,\kappa_\perp(p)} \, \P_\perp(\n) \right]_{y=\n\cdot \x} \, \Lambda(\dd \n) = \\
 -\frac{1}{2 (2 \uppi)^3\, \rho} \int_\mathcal{B} \frac{\dd p}{p^2} 
\int_{\mathcal{S}_+} \left[ \kappa_\parallel(p)^3\,  \e^{-\n\cdot\x \,\kappa_\parallel(p)} \, \P_\parallel(\n) + 
\kappa_\perp(p)^3\,  \e^{-\n\cdot\x \,\kappa_\perp(p)} \, \P_\perp(\n) \right] \, \Lambda(\dd \n)
\end{multline*}
Integration over $\mathcal{S}_+$ can be  explicitly carried out:
\begin{multline*}
\int_{\mathcal{S}_+}  \e^{-\n\cdot\x \, \kappa_\parallel(p) } \,\P_\parallel(\n) \, \Lambda(\dd \n) = 
\kappa_\parallel(p)^{-2} \, \nabla \nabla^\top \int_{\mathcal{S}_+}  \e^{-\n\cdot\x \, \kappa_\parallel(p) } \, \Lambda(\dd \n) = \\
\frac{2 \uppi}{\kappa_\parallel(p)^2} \, \nabla \nabla^\top \, \int_0^1 \e^{-r z \, \kappa_\parallel(p)} \, \dd z =
\frac{2 \uppi}{\kappa_\parallel(p)^3} \, \nabla \nabla^\top \, \frac{\e^{-r \, \kappa_\parallel(p)}}{r} 
\end{multline*}
and
$$\int_{\mathcal{S}_+} \e^{- \n\cdot\x\, \kappa_\perp(p)} \, \Lambda(\dd \n) = \frac{2 \uppi}{\kappa_\perp(p)}\,
\left[\I - \kappa_\perp(p)^{-2} \nabla \nabla^\top\right] \frac{\e^{-r \, \kappa_\perp(p)}}{r} $$

We now note that the CBF functions $\kappa_\parallel(p)$ and $\kappa_\perp(p)$ can be expressed in the form
\begin{gather}
\kappa_\parallel(-\ii \omega) = -\ii \omega/c_\parallel(\omega) + \mathcal{A}_\parallel(\omega)\\
\kappa_\perp(-\ii \omega) = -\ii \omega/c_\perp(\omega) + \mathcal{A}_\perp(\omega)
\end{gather}
where $r := \vert \x \vert$, $c_\parallel(\omega)$ and $c_\perp(\omega)$ represent the phase speeds of the longitudinal and transverse modes, while 
$\mathcal{A}_\parallel(\omega),\mathcal{A}_\perp(\omega) \geq 0$ are the corresponding attenuation functions. 

Changing the integration variable according to the formula  $p = -\ii \omega$ 
\begin{multline} \label{eq:last}
\G(t,\x) = -\frac{1}{2 (2 \uppi)^3\, \rho} \\ \times \int_{-\infty}^\infty \Big\{ \left[\frac{\rho}{\overline{\lambda}(-\ii \omega) 
+ 2 \overline{\mu}(-\ii \omega)}\right]^{1/2} \,
\kappa_\parallel(-\ii \omega)^{-2} \, \nabla \nabla^\top \, \left[r^{-1} \, \e^{-\ii \omega (t - r/c_\parallel(\omega)) - r\,
\mathcal{A}_\parallel(\omega)}\right] \\ +  \left[\frac{\rho}{\overline{\mu}(-\ii \omega)}\right]^{1/2}
\left[\I - \kappa_\perp(-\ii \omega)^{-2} \nabla \nabla^\top\right] \, \left[r^{-1} \, \e^{-\ii \omega (t - r/c_\perp(\omega)) - r\,
\mathcal{A}_\perp(\omega)}\right] \Big\} \, \dd \omega  
\end{multline}

The phase speeds tend to the values $c_\parallel^0 := \sqrt{(\lambda(0) +2 \mu(0))/\rho}$ and 
$c_\perp^0 := \sqrt{\mu(0)/\rho}$ for $\omega \rightarrow \infty$ and $c_\parallel(p)^{-1} = 1/c_\parallel^0 + o[1]$, 
$c_\perp(p)^{-1} = 1/c_\perp^0 + o[1]$. The numbers $c_\parallel^0$ and $c_\perp^0$ constitute the upper bounds on
the corresponding wave speeds. Using Jordan's Lemma it is then possible to prove that 
the longitudinal wavefield, represented by the first integral in \eqref{eq:last}, vanishes for $t < r/c_\parallel^0$ and the transverse field, represented by the second integral, vanishes for $t < r/c_\perp^0$ \cite{HanWM2013}. The spheres $r/c_\parallel^0 - t = 0$
and $r/c_\perp^0 - t = 0$ constitute the wavefronts of $\G(t,\x)$.

Regularity of Green´s function $\G(t,\x)$ at the wavefronts is controlled by the behavior of the attenuation function at infinity
\cite{HanJCA,HanUno,HanDue}. For example, if $\mathcal{A}_\parallel(\omega) \rightarrow \mathcal{A}_\parallel^\infty$ for $\omega \rightarrow \infty$ 
then the wavefront of the longitudinal wave can carry a finite jump, while
if $\mathcal{A}_\parallel(\omega) = \OO_\infty\left[ a\, \omega^\alpha\right]$, with $0 < \alpha < 1$, $a > 0$, then the longitudinal wavefield
is infinitely smooth at the wavefront.

\section{An example of an anisotropic medium - the TI medium}
\label{sec:TI}

In a general anisotropic medium an eigenvalue of the matrix $\K_\n(p)$, 
$p \in \mathbb{R}_+$, need not be a CBF except if the corresponding eigenvector 
is constant.

\begin{theorem}  \label{thm:except}
If a unit eigenvector $\f$ of $\K_\n(p)$, $p \in \mathbb{R}_+$, is independent of $p$, then 
the corresponding eigenvalue $\kappa(p)$ is a CBF.
\end{theorem}
\begin{proof}
Using Theorem~\ref{thm:CBF} 
$$\kappa(p) =  \f^\top \, \K_\n(p)\, \f = 
\f^\top \, \B_\n\, \f + p \int_{]0,\infty[} \frac{\f^\top\, \M_\n(r) \, \f}{p + r} \mu(\dd r)$$
with $\f^\top \, \B_\n\, \f \geq 0$ and 
$\f^\top\, \M_\n(r) \, \f \geq 0$ $\mu-$almost everywhere, hence 
by the same theorem, $\kappa$ is a CBF. \smartqed
\end{proof} 

The simplest possible example of an anisotropic medium - the transversely isotropic medium
with the symmetry axis $\ee$ independent of the inverse relaxation time $r$ -- illustrates the problems associated with anisotropy.

In a coordinate system in which the symmetry axis $\ee = [0,0,1]^\top$ 
the tensor $\Q(p)$ can be expressed in the Voigt notation as 
$$\Q(p) = \left[ \begin{array}{cccccc}
q_{11}(p) & q_{12}(p) & q_{13}(p) & 0 & 0 & 0 \\
q_{12}(p) & q_{11}(p) & q_{13}(p) & 0 & 0 & 0 \\
q_{13}(p) & q_{13}(p) & q_{33}(p) & 0 & 0 & 0 \\
0 & 0 & 0 & q_{44}(p) & 0 & 0 \\
0 & 0 & 0 & 0 & q_{44}(p) & 0 \\
0 & 0 & 0& 0 & 0 & q_{66}(p)
\end{array}\right]
$$
where $q_{12}(p) = q_{11}(p) - 2 q_{66}(p)$ \cite{Fedorov}. 

It is assumed that $\Q(p)$ is a CBF, hence in particular the diagonal elements $q_{11}(p), q_{33}(p), q_{44}(p)$ and $q_{66}(p)$ are CBFs (use Theorem~\ref{thm:CBF}). The matrix-valued function $\Q(p)$ and its component functions are defined on the cut complex $p$-plane.

Many formulae known from the theory of elastic waves in anisotropic media carry over to the viscoelastic case by replacing the stiffness coefficients $c_{ijkl}$ by 
the functions $q_{klmn}(p)$. We shall therefore use the book of Fedorov 
\cite{Fedorov} as a source of such formulae. 
 
The matrices $\Q_\n(p)$ and $\Q_\n(p)^{-1}$ can be expressed as linear combinations
of the matrices $\I$, $\n\, \n^\top$, $\ee \,\ee^\top$ and $\f \,\f^\top$, where 
$\f = \left(1 - (\ee\cdot \n)^2\right)^{-1/2}\, \ee \times \n$ (a vector product)
with coefficients which are algebraic functions of $q_{kl}(p)$ (\cite{Fedorov},
Sec.~32). By analytic continuation these expressions are defined on the cut complex $p$-plane. The vector $\f$ is a unit eigenvector of $\Q_\n(p)$. It defines the polarization of the transverse wave. The corresponding eigenvalue $\kappa^{\mathrm{T}}(p)$ is a CBF.
Consequently its analytic continuation to the imaginary axis has the form
$$\kappa^{\mathrm{T}}(-\ii \omega) = -\ii \omega/c^{\mathrm{T}}(\omega) 
+ \mathcal{A}^{\mathrm{T}}(\omega)$$ with $\mathcal{A}^{\mathrm{T}}(\omega) \geq 0$. The 
functions $\kappa^{\mathrm{T}}(p)$, $\mathcal{D}^{\mathrm{T}}(\omega)$ 
and $\mathcal{A}^\mathrm{T}(\omega)$ enjoy all the properties of the corresponding 
functions in the scalar viscoelasticity \cite{HanWM2013}. We note that
$$q^\mathrm{T}(p) = q_{66}(p) \, \left(1 - (\ee\cdot\n)^2\right) + q_{44}(p)\,
(\ee \cdot\n)^2$$
and $\kappa^{\mathrm{T}}(p) = \rho^{1/2}\,p/q^{\mathrm{T}}(p)$.

Let $\K^\prime_\n(p)$ denote the matrix $\K_\n(p)$ in the adjusted coordinate system 
chosen in such a way that $\f^\prime = [0,0,1]^\top$, $\K^\prime_\n(p) =
\mathbf{R}\, \K_\n(p)\, \mathbf{R}^{-1}$, where $\mathbf{R}$ is a rotation.
We then have 
$$\K^\prime_\n(p) = \left[ \begin{array}{cc} \K^\perp_\n(p)  & 0 \\ 
0 & \kappa^{\mathrm{T}}(p)  \end{array} \right] $$
where $\K^\perp(p)$ denotes the projection of $\K^\prime_\n(p)$ onto the plane spanned by the vectors $\n$ and $\ee$.  We now note that the matrix-valued
function $\K^\prime_\n(p)$ is a CBF. Indeed, $\mathbf{R}^\dag = \mathbf{R}^{-1}$, hence 
\begin{multline*}
\Im \K^\prime_\n(p) = \frac{1}{2 \ii} \left( \mathbf{R}\, \K_\n(p) \, \mathbf{R}^{-1}
- \left[\mathbf{R}\, \K_\n(p) \, \mathbf{R}^{-1}\right]^\dag\right) = \\
 \frac{1}{2 \ii} \mathbf{R}\, \left[\K_\n(p) - \K_\n(p)^\dag\right] \, \mathbf{R}^{-1}
= \mathbf{R}\, \Im \K_\n(p) \mathbf{R}^{-1}
\end{multline*}
But $\Im p \, \Im \K_\n(p) \geq 0$, hence $\Im p \,\Im \K^\prime_\n(p) \geq 0$, q. e. d.

$\kappa^{\mathrm{T}}(p)$ is a CBF, hence 
$\Im p \, \Im \kappa^{\mathrm{T}}(p) \geq 0$.
Consequently $\Im p \, \Im \K^\prime_\n(p) \geq 0$ if and only if $\Im p \, \Im \K^\perp_\n(p) \geq 0$. 
Since $\K^\prime_\n(p)$ is a CBF, this implies that $\K^\perp_\n(p)$ is a $2\times 2$ matrix-valued CBF. 

The theory developed above for matrix-valued $\K_\n(p)$ can now be applied to $\K^\perp_\n(p)$. All the formulae remain valid in the cut complex $p$-plane, in particular on 
the imaginary axis $p = -\ii \omega$. The theory of matrix-valued attenuation and dispersion developed in the
previous sections can now be applied to $\K^\perp_\n(p)$.

\begin{lemma}
If the $3 \times 3$ complex matrix
$$\A = \left[ \begin{array}{cc}
\B & 0 \\
0 & b \end{array} \right]$$
where $\B$ is a $2 \times 2$ complex matrix, $b \in \mathbb{C}$ and $f$ is an analytic function, then
\begin{equation} \label{eq:dec}
f(\A) = \left[ \begin{array}{cc} f(\B) & 0 \\ 0 & f(b) \end{array} \right]
\end{equation}
\end{lemma}
\begin{proof}
Note that
$$(s \I - \A)^{-1} = \left[ \begin{array}{cc} (s\, \I_2 - \B)^{-1} & 0 \\
0 & (s - b)^{-1}\end{array}\right]$$
where $\I_2$ denotes the $2\times 2$ unit matrix and $\Gamma$ encircles the 
spectrum of $\A$ in the positive direction
Hence 
\begin{multline*}
f(\A) = \frac{1}{2 \uppi \ii} \int_\Gamma (s \I - \A)^{-1}\, f(s) \, \dd s = \\
\frac{1}{2 \uppi \ii}\left[ \begin{array}{cc}  \int_\Gamma (s \I_2 - \B)^{-1}\, f(s)\, \dd s & 0 \\
0 & \int_\Gamma (s - b)^{-1}\, f(s) \,\dd s\end{array} \right]
\end{multline*}
The spectrum of $\A$ consists of the spectrum of $\B$ and the number $b$.
Hence equation~\eqref{eq:dec} follows. \smartqed
\end{proof}

From the above lemma we conclude that
$$\e^{-\K_\n(p) \, y} = \left[ \begin{array}{cc} \exp(-\K^\perp_\n(p)\, y) & 0 \\
0 & \exp(-\kappa^{\mathrm{T}}(p) \, y) \end{array}\right]$$
The factor $\K_\n(p)^{-1}\, \Q_\n(p)^{-1}$ in \eqref{eq:planewave} can similarly be recast in a block matrix form.
Consequently every plane wave is represented in the block form in a coordinate system in which $\f = [0,0,1]^\top$. In the adjusted coordinate system 
the attenuation and dispersion of the quasi-longitudinal and quasi-transverse 
plane wave are represented by $2\times 2$ matrix-valued functions. 

Before summing over the plane-wave wavefront normals $\n$ the expressions obtained above 
must be rotated to a fixed coordinate system. Suppose again that in this coordinate system
$\ee = [0,0,1]^\top$. A coordinate system in which the matrices $\K_\n(p)$
and $\Q_\n(p)$ have block form is obtained by a rotation $\mathbf{R}$ which transforms
the vector $\f$ to $\ee$: $\mathbf{R} \f = \ee$. A rotation around the axis $\mathbf{a} := 
\ee \times \f$ through $90^{\mathrm{o}}$ in the positive direction satisfies this condition. The opposite rotation $\mathbf{R}^{-1}$ restores the original coordinate system.
Thus in the original coordinate system
\begin{equation} \label{eq:last}
\K_\n(p) = \mathbf{R}^{-1} \, \left[ 
\begin{array}{cc} \K^\perp_\n(p) & 0 \\ 0 & \kappa^{\mathrm{T}}(p) \end{array} \right] \, \mathbf{R} 
\end{equation}

From the last result we can derive the block-form expressions for phase speed, attenuation and dispersion. If $\B^\perp_\n := \lim_{p\rightarrow \infty} \K^\perp_\n(p)/p$ and 
$b^{\mathrm{T}} := \lim_{p\rightarrow\infty} \kappa^{\mathrm{T}}(p)/p$, then
$$\B_\n = \mathbf{R}^{-1} \, \left[ 
\begin{array}{cc} \B^\perp_\n & 0 \\ 0 & b^{\mathrm{T}} \end{array} \right] \, \mathbf{R}$$
Similarly
$$\A_\n(\omega) = \mathbf{R}^{-1} \, \left[ 
\begin{array}{cc} \A^\perp_\n(\omega) & 0 \\ 0 & a^{\mathrm{T}}(\omega) \end{array} \right] \, \mathbf{R},$$
where $\A^\perp_\n(\omega) := \Re \K^\perp_\n(-\ii \omega) \geq 0$, 
$a^{\mathrm{T}}(\omega) := \Re \kappa^{\mathrm{T}}(-\ii \omega) \geq 0$
and
$$\mathbf{D}_\n(\omega) = \mathbf{R}^{-1} \, \left[ 
\begin{array}{cc} \mathbf{D}^\perp_\n(\omega) & 0 \\ 0 & d^{\mathrm{T}}(\omega) \end{array} \right] \, \mathbf{R}$$. 

In viscoelastic media with lower symmetry the eigenvectors of the matrix $\K_\n(-\ii \omega)$ 
in general depend on frequency except for waves propagating along symmetry axes. This entails
that attenuation couples all the modes for wavefront normals not directed along acoustic axes. 

In the case of the wavefront normal $\n$ directed along a symmetry axis 
$L^2$ in a monoclinic medium \cite{Fedorov} one of the three polarizations  
in an elastic medium is parallel to the axis, while two other polarizations lie in the mirror symmetry 
plane (the plane orthogonal to the symmetry axis). The directions of the latter depend on 
the stiffness tensor. In a viscoelastic medium the stiffness tensor depends on frequency 
and as a result the directions of the two transverse waves also depend on frequency. This entails that  the corresponding eigenvalues of the matrix $\K_\n(p)$ are not complete Bernstein functions. Consequently the two transverse waves are in general coupled by the attenuation function. 
In the case of a symmetry axis of a higher order $L^k$, $k > 2$, the entire symmetry plane is an eigenspace of $\K_\n(p)$ 
and transverse polarizations can lie anywhere in that plane. In this case the dependence on frequency disappears and the two transverse waves are not coupled. 

\section{Concluding remarks.}

Green's function is a superposition of plane waves. The analysis of its structure 
can thus be reduced to the analysis of the plane waves.

Anisotropic effects in elastic plane wave propagation are represented by a single tensor, viz. the inverse phase speed tensor (or, equivalently, by the acoustic tensor). The 
eigenvectors of these tensors represent the polarizations of three independently propagating modes.

Viscoelastic plane waves additionally involve the attenuation function which is tensor-valued and does not in general commute with the inverse phase speed tensor. The
attenuation function couples the three modes defined by the eigenvectors of the inverse phase speed tensor. In order to account for attenuation the three modes of 
a plane wave have to be considered jointly. Only 
those modes whose polarization vectors are independent of frequency can be decoupled 
from the other modes in an adjusted coordinate system.  

In the isotropic case all the polarization vectors are independent of frequency and the 
longitudinal mode and the transverse mode decouple so that Green's function can be expressed as a superposition of two independent spherical waves. The two waves have 
many properties of the scalar viscoelastic waves. Regularity of the wavefields at their respective wavefronts can be studied by the methods of \cite{HanJCA}.

The theory of matrix-valued CBFs outlined in this paper allows an identification
of the matrix-valued attenuation function and a deep analysis of its properties. 
To our best knowledge such a theory has not been developed before.  
Our construction of matrix-valued attenuation is a generalization of the method presented in \cite{HanWM2013}
in the context of scalar viscoelastic wave propagation. In case of need the 
anisotropic aspects of the  high- and low-frequency asymptotics of the attenuation function.

Applications of our representation of the viscoelastic wave field depend on an efficient, accurate and stable evaluation of matrix exponential. The last topic is discussed in detail in \cite{MolerVanLoan}.

\appendix  

\section{Some functions of matrix arguments.}
\label{app:matrices}

If $f$ is a (scalar) CBF then there are two non-negative real numbers $a, b$ and a positive Radon 
measure $\rho$ on $\mathbb{R}_+$ such that
$$f(x) = a + b \, x + \int_{]0,\infty[} \frac{x}{s + x} \rho(\dd s)$$
Using this formula the function $f$ can be extended to $\mathcal{M}^\mathbb{C}$:
\begin{equation} \label{eq:defbyCBF}
f(\B) := a \, \I + b \, \B + \int_{]0,\infty[} \B\, (s \, \I + \B)^{-1} \, \rho(\dd s)
\end{equation} 
 \cite{Schilling11}. 
Schilling has proved the following statement:\\
If $f, g$ and their pointwise product $f\, g$ are CBFs then 
$(f\, g)(\B) = f(\B)\, g(\B)$ for every matrix $\B \in \mathcal{M}^\mathbb{C}$.

\begin{lemma}
\begin{equation} \label{eq:a1}
x^\alpha = \frac{\sin(\alpha\,\uppi)}{\uppi} \int_0^\infty \frac{x \, s^{\alpha-1}}{x + s} 
\dd s \qquad \text{for $x \geq 0$}
\end{equation}
\end{lemma}
\begin{proof}
Identity~\eqref{eq:a1} is equivalent to the identity
\begin{equation} \label{eq:a2}
x^{\alpha-1} = \frac{\sin(\alpha\,\uppi)}{\uppi} \int_0^\infty \frac{s^{\alpha-1}}{x + s} 
\dd s \qquad\text{ for $x \geq 0$}
\end{equation} 

In order to prove the last identity we use the fact that the Stieltjes transform is an iterated 
Laplace transform and
$$\int_0^\infty s^{\alpha-1} \, \e^{-y s}\, \dd s = y^{\alpha-1} \, \int_0^\infty 
z^{\alpha-1}\, \e^{-z} \, \dd z = \Gamma(\alpha)\, y^{\alpha-1}$$
$$\Gamma(\alpha) \int_0^\infty y^{-\alpha} \, \e^{-x y} \, \dd y = \Gamma(\alpha)\,
\Gamma(1-\alpha) \, x^{\alpha-1} = \frac{\uppi}{\sin(\alpha\, \uppi)} x^{\alpha-1}$$ 
which proves \eqref{eq:a2}. \smartqed
\end{proof} 

This provides us with an integral definition of the square root $\B^{1/2}$:
\begin{equation} \label{eq:a3}
\B^{1/2} = \uppi \int_0^\infty \B\, (s\,\I + \B)^{-1}\, 
s^{-1/2}\, \dd s
\end{equation}
provided that $\B$ has no non-positive real eigenvalue.
We now note that the functions $f(x) = g(x) = x^{1/2}$ and $f\, g$ are CBFs, 
and the definition \eqref{eq:a3} has the form of equation~\eqref{eq:defbyCBF},
hence
\begin{equation} \label{eq:sqroot} 
\B^{1/2} \, \B^{1/2} = \B
\end{equation}
 as expected. 

The square roots of a matrix $\B$ are usually defined as solutions $\mathbf{Y}$ of the equation
\begin{equation} \label{eq:rooteq}
\mathbf{Y}^2 - \B = 0
\end{equation}
If a matrix $\B \in \mathcal{M}_d^\mathbb{C}$ has no non-positive real eigenvalue
then it has a unique square root with the property that its eigenvalues lie in the
open right half of the complex plane \cite{Higham87}. This particular square root is called the 
principal square root.
An algorithm for numerical calculation of the principal square root can be found in \cite{Meini}.

\begin{lemma} \label{lem:sqr}
If $\Re \B > 0$  
then the square root $\mathbf{Y} = \B^{1/2}$ defined by equation~\eqref{eq:a3} is the principal 
square root. 
\end{lemma}
\begin{proof}
Equation~\eqref{eq:sqroot} implies that $\mathbf{Y}$ satisfies equation~\eqref{eq:rooteq}. The real part of the integrand of \eqref{eq:a3} can be expressed in the form 
\begin{multline*}
\frac{1}{2} \left[\B\, \left( \B + s\, I\right)^{-1} +  \left(\B^\dag + s \, \I\right)^{-1} \, \B^\dag\right] =\\
\frac{1}{2} \left(\B^\dag + s \, \I\right)^{-1}\, \left[ \B^\dag \, (\B + s\, \I) + \left(\B^\dag + s \, \I\right) \, \B\right]\,
(\B + s\, \I)^{-1} =  \mathbf{U}^\dag \,\left[\B^\dag\, \B + s\, \Re \B\right]\,\mathbf{U}
\end{multline*}
where $\mathbf{U} := (\B + s \, \I)^{-1}$. Hence for $s > 0$ the real part of the integrand of \eqref{eq:a3} is positive
 and therefore the square root defined by \eqref{eq:a3} satisfies the inequality $\Re \B^{1/2} > 0$.
This result along with equation~\eqref{eq:sqroot} implies that the right-hand side of \eqref{eq:a3}
is the principal square root.  \smartqed
\end{proof} 

\begin{theorem} \label{lem:derexp}
$$\frac{\dd}{\dd s} \e^{s \A} = \A \, \e^{s \A}$$
\end{theorem}
\begin{remark}
It follows easily from equation~\eqref{eq:Cauchy} that $\A$ commutes with $\e^{s \A}$.
\end{remark}
\begin{proof}
\begin{multline*}
\frac{\dd}{\dd s} \e^{s \A}  = \frac{1}{2 \uppi \ii}  \frac{\dd}{\dd s} \e^{s \A}  \int_\Gamma \left[ s \, \I - \A\right]^{-1}\, \e^s\, \dd s = 
\frac{1}{2 \uppi \ii}  \,  \int_\Gamma s\, \left[ s \, \I - \A\right]^{-1}\, \e^s \, \dd s = \\
\frac{1}{2 \uppi \ii}  \,  \int_\Gamma (s\, \I - \A + \A) \, \left[ s \, \I - \A\right]^{-1}\, \e^s \, \dd s = 
\frac{1}{2 \uppi \ii}  \,  \int_\Gamma  \e^s \, \dd s + \A \, \e^{s \A}
\end{multline*}
The first term vanishes. \smartqed
\end{proof}

For the sake of convenience we have chosen the definition \eqref{eq:Cauchy} 
of an analytic function of a matrix. We shall show that this definition is equivalent to the more common definition in terms of the Maclaurin series.

We shall begin with proving that $f_n(\A) = \A^n$, $n \in \mathbf{N}$, where
\begin{equation}
f_n(\A) := \frac{1}{2 \uppi \ii} \int_\Gamma s^n \, (s \, \I - \A)^{-1} \, \dd s
\end{equation}
and $\A^n := \A \ldots \A$ ($n$-th power of $\A$). Note first that 
if $g(s) = s \, f(s)$ then $g(\A) = \A \, f(\A)$:
\begin{multline*}
g(\A) = \frac{1}{2 \uppi \ii} \int_\Gamma s\cdot f(s)\, (s\, \I - \A)^{-1} \, \dd s =
\\
\frac{1}{2 \uppi \ii} \int_\Gamma f(s) \, \dd s + \A \frac{1}{2 \uppi \ii} \int_\Gamma
f(s)\, (s\, \I - \A)^{-1} \, \dd s = \A \, f(\A)
\end{multline*}
For $f(s) \equiv 1$ this yields the formula $f_1(\A) = \A$, while from $f_{n+1}(s) = s\, f_n(s)$ we obtain the recursive relation 
$\f_{n+1}(\A) = \A \, f_n(\A)$. By induction, $f_n(\A) = \A^n$ for $n \in \mathbb{N}$. 

Suppose that $f(s) = \sum_{n=1}^\infty a_n \, s^n$. Substituting this series expansion 
in \eqref{eq:Cauchy} and using the result of the previous paragraph, we obtain
the formula
$$f(\A) = \sum_{n=0}^\infty a_n\, \A^n$$
The series converges because it is majorized by $f\left(\| \A \|_d\right)$.

\section{Some properties of matrix-valued CBFs and Stieltjes functions.}
\label{app:CBFs}

\begin{theorem} \label{thm:a1}
If $\B \in \mathcal{M}^\mathbb{C}_d$ and $\Im \B \geq 0$ then the square root defined by \eqref{eq:a3} satisfies the inequality 
$\Im \B^{1/2} \geq 0$. 
\end{theorem} 
\begin{proof}
\begin{multline*}
\frac{1}{2 \ii} \left[ \B \, (\B + s \I)^{-1} - \left(\B^\dag + s\, \I\right)^{-1}\, \B^\dag\right] = \\
\frac{1}{2 \ii} \left(s\, \I + \B^\dag \right)^{-1}\, \left[\left(s \, \I + \B^\dag\right)\, \B
 - \B^\dag (\B + s\, \I)\right] \, (\B + s\, \I)^{-1} = 
s \, \mathbf{U}^\dag\, \Im \B \, \mathbf{U} \geq 0 
\end{multline*}
where $\mathbf{U} := (s\, \I + \B)^{-1}$. In view of equation~\eqref{eq:a3} this implies that 
$\Im \B^{1/2} \geq 0$. \smartqed
\end{proof}

In view of Definition~\ref{def:CBF} this entails an important corollary:
\begin{corollary}\label{cor:half}
If $\A(x)$ is a matrix-valued CBF then the principal square root $\A(x)^{1/2}$ is a CBF.
\end{corollary}

\begin{theorem} \label{thm:CBF}
If the matrix-valued function $\A(x)$ on $\mathbb{R}_+$ is a CBF then there are two positive semi-definite 
matrices $\B$ and $\C$, a positive Radon measure $\mu$ on $\mathbb{R}_+$ satisfying the
inequality~\eqref{eq:6} and a
measurable $\mu$-almost everywhere bounded and positive semi-definite function $\mathbf{M}$
such that 
\begin{equation} \label{eq:5}
\A(x) = \B + x\, \C + x \int_{]0,\infty[} (x + s)^{-1} \mathbf{M}(s)\, \mu(\dd s), \qquad x \geq 0
\end{equation}
\end{theorem}
\begin{proof}
If $\B \in \mathcal{M}^\mathbb{C}_d$ and $\Im \B \geq 0$ then for every $\v \in
\mathbb{C}^d$ 
$$\Im \v^\dag\, \B \, \v = \frac{1}{2 \ii} \left[ \v^\dag \, \B \, \v - 
\overline{\v^\dag \, \B \, \v}\right] = \frac{1}{2 \ii} \v^\dag\, \left(\B - \B^\dag\right)\, 
\v \geq 0$$

If $\A$ is a CBF then 
$\Im \left[\v^\dag\, \A(z)\, \v\right] \geq 0$ for every $\v \in \mathbb{C}^d$ and $\Im z > 0$.
Furthermore $\lim_{x\rightarrow 0+} \A(x)$ exists and is real. 
For real $x$ the matrix $\A(x)$ is real, hence it is sufficient to consider $\v \in \mathbb{R}^d$.
It follows that  $\v^\top\, \A(x)\, \v$ is a CBF for every $\v \in \mathbb{R}$ and 
therefore there are two non-negative numbers $b_\v, c_\v$ and a positive measure $\mu_\v$ satisfying 
inequality~\eqref{eq:6} such that 
$$ \v^\top\, \A(x)\, \v = b_\v + x \, c_\v + x \int_{]0,\infty[} (x + r)^{-1} \, \mu_\v(\dd r)$$
\cite{BernsteinFunctions}. It is clear that $b_\v = \v^\top \, \B\, \v$, where $\B := \A(0)$.
As for $c_\v$,
$$c_\v = \lim_{x\rightarrow\infty} \left[ x^{-1}\v^\top\, \A(x)\, \v\right]$$
This proves that the limit on the right exists, hence a symmetric matrix $\C$ can be defined by polarization
$$2 \v^\top \, \C \, \w = c_{\v + \w}- c_\v - c_\w $$
It is easy to see that the right-hand side is a linear function of $\v$ and $\w$ and therefore it defines a symmetric positive semi-definite matrix $\C$. 

Let $\A_0(x) := \A(x) - \B - x \, \C$,
$$\v^\top\, \A_0(x)\, \v = x \int_{]0,\infty[} (x + s)^{-1}\, \mu_\v(\dd s)$$
By polarization 
$$\v^\top\, \A_0(x)\, \w = \int_{]0,\infty[} (x + s)^{-1}\, N_{\v,\w}(\dd r)$$
where $2 N_{\v,\w}(I) := \mu_{\v+\w}(I) - \mu_\v(I) - \mu_\w(I)$ for every interval 
$I \subset \mathbb{R}_+$. $N_{\v,\w}(I)$ can be expressed in terms of
$\v^\top\,\A_0\, \w$:
$$N_{\v,\w}(I) = \lim_{\varepsilon\rightarrow 0+} \frac{1}{\uppi} \int_I \Im 
\frac{\v^\top\, \A_0(-s + \ii \varepsilon) \, \w}{s - \ii \varepsilon} \, \dd r$$
for every segment $I$ whose ends are continuity points of the measure $N_{\v,\w}$
\cite{BernsteinFunctions}. 
The last expression shows existence of a $\mathcal{M}_d$-valued measure $\N$ such that 
$N_{\v,\w}(I) = \v^\top\, \N(I)\, \w$. $\N(I)$ is positive semi-definite because
for every $\v \in \mathbb{R}_+$ $\v^\top\, \N(I)\, \v \geq 0$.

Now 
$$\v^\top\, \N(I) \, \v + \w^\top \, \N(I)\, \w \pm 2 \v^\top\, \N(I)\, \w
= (\v \pm \w)^\top\, \N(I)\, (\v \pm \w) \geq 0$$
hence
$$\vert \v^\top\, \N(I) \, \w  \vert \leq \frac{1}{2} \left[ \v^\top\, \N(I)\, \v
+ \w^\top \, \N(I) \, \w\right] \leq \frac{1}{2} \left(\| \v \|^2 +  \, \| \w\|^2\right) \, \mu(I)$$
where $\mu(I) := \mathrm{trace}[\N(I)]$. The Radon measure $\mu$ thus defined on
$\mathbb{R}_+$ is positive. The measures $\mu_\v$ satisfy the inequality~\eqref{eq:6},
hence $\N_\v$ and $\mu$ satisfy the same inequality.

By the Radon-Nikodym theorem \cite{Rudin76} there is a $\mu$-almost everywhere bounded function
$\M(r)$ on $\mathbb{R}_+$ such that $\N(\dd r) = \M(r)\, \mu(\dd r)$. \smartqed
\end{proof}

\begin{corollary} \label{cor:calc}
If $\A(x)$ is a matrix-valued CBF satisfying~\eqref{eq:5}
$$\N(]a,b]) = \frac{1}{\uppi} \lim_{\varepsilon\rightarrow 0+}\int_{]a,b]} 
\Im \left[\frac{\A_0(-s + \ii \varepsilon)}{s - \ii \varepsilon}\right] \dd s$$
and $\mu(]a,b]) = \mathrm{trace}[\N(]a,b])]$
for every regular point $b > 0$ of the measure $\N$ and every $a > 0$,
where $\A_0(x) = \A(x) - \B - x \, \C$. 
\end{corollary}

Equation~\eqref{eq:5} provides an analytic continuation of the function $\A(x)$ to the 
complex plane cut along the negative real semi-axis. 

The function $\M(s)$ can be assumed bounded by 1 almost everywhere in the sense of measure $\mu$.

\begin{corollary} \label{cor:coeffs}
If $\A(x)$ satisfies equation~\eqref{eq:5} and \eqref{eq:6}, then
\begin{enumerate}[(i)]
\item $\A(0) = \B$;
\item $\C = \lim_{x\rightarrow \infty} x^{-1}\, \A(x) \quad \text{for $x \in \mathbb{R}_+$}$.
\end{enumerate}
\end{corollary}
\begin{proof}
(i) follows from equation~\eqref{eq:5};

Concerning (ii), it is sufficient to prove that
$\mathbf{L}(x) := \int_{]0,\infty[} (x + s)^{-1} \M_\n(s) \, \mu(\dd s) \rightarrow 0$
for $x \rightarrow \infty$.
Indeed, for $ x \geq 1$,
$$\vert \mathbf{L}(x) \vert \leq \int_{]0,\infty[} (x + s)^{-1}\, \mu(\dd s) \leq 
\int_{]0,\infty[} (1 + s)^{-1}\, \mu(\dd s) < \infty$$
on account of \eqref{eq:6}. The Lebesgue Dominated Convergence Theorem implies that 
$\lim_{x\rightarrow\infty} \mathbf{L}(x) = 0$, q.e.d. \smartqed
\end{proof} 

For the analytic continuation $\A(z)$, $z \in \mathbb{C}\setminus ]-\infty,0[$, the corollary can be
easily extended
to the limit $\Re z \rightarrow \infty$. 

\begin{theorem}\label{thm:S}
If the matrix-valued function $\A(x)$ is a Stieltjes function then there are two positive 
semi-definite 
matrices $\B$ and $\C$, a positive Radon measure $\mu$ on $\mathbb{R}_+$ satisfying the
inequality~\eqref{eq:6} and a
measurable $\mu$-almost everywhere bounded and positive semi-definite function $\mathbf{M}$
such that 
\begin{equation} \label{eq:S}
\A(x) = \B + x^{-1}\, \C + \int_{]0,\infty[} (x + s)^{-1} \mathbf{M}(s)\, \mu(\dd s)
\end{equation}
\end{theorem}

The proof of this theorem, based on the integral representation 
$$f(x) = a + b/x + \int_{]0,\infty[} (x + s)^{-1} \, \mu(\dd s)$$
\cite{BernsteinFunctions}
with $a, b \geq 0$ and a positive Radon measure $\mu$ satisfying equation~\eqref{eq:6},
is analogous to the proof of Theorem~\ref{thm:CBF}. 

Equation~\eqref{eq:S} provides an analytic continuation of the function $\A(x)$ to the complex plane cut along the negative real semi axis. 

Comparison of Theorems~\ref{thm:CBF} and \ref{thm:S} yields the following corollary
\begin{corollary} \label{cor:toS}
If $\A$ is a matrix-valued CBF then $x^{-1}\, \A(x)$ is a matrix-valued Stieltjes
function.
\end{corollary}

\end{document}